\documentclass{article}
\usepackage{amsmath,amssymb,amsthm,setspace}

\usepackage{fullpage}
\usepackage[utf8]{inputenc}
\usepackage[OT4]{fontenc}
\usepackage{authblk}
\usepackage{enumerate}
\newtheorem{theorem}{Theorem}
\newtheorem{lemma}[theorem]{Lemma}

\newtheorem{corollary}[theorem]{Corollary}

\usepackage[hidelinks]{hyperref}

\theoremstyle{remark}
\newtheorem{example}[theorem]{Example}

\newcommand{\Exp}[1]{\operatorname{E}[#1]}
\newcommand{\Prb}[1]{\operatorname{P}[#1]}
\newcommand{\per}{\operatorname{per}}
 
 \newcommand{\Sb}{\bar{S}}
 \newcommand{\Sd}{S_{\$}}

\usepackage{todonotes}

\title{Maximal Unbordered Factors of Random Strings\thanks{A preliminary version of this paper~\cite{cording2016maximal} with weaker results was presented at the 23rd Symposium on String Processing and Information Retrieval (SPIRE '16).}}
 \author[1]{Patrick Hagge Cording\thanks{Supported by the Danish Research Council under the Sapere Aude Program (DFF 4005-00267).}}
 \author[2]{Travis Gagie\thanks{Supported by FONDECYT grant 1171058.}}
 \author[3]{Mathias B{\ae}k Tejs Knudsen\thanks{Partly supported by Mikkel Thorup's Advanced Grant from the Danish
 Council for Independent Research under the Sapere Aude research career programme and the FNU project AlgoDisc --- Discrete Mathematics, Algorithms, and Data Structures.}}
 \author[4]{Tomasz~Kociumaka}
\affil[1]{DTU Compute, Technical University of Denmark, \texttt{patrick.cording@gmail.com}}
\affil[2]{CeBiB; EIT, Universidad Diego Portales, Chile, \texttt{travis.gagie@gmail.com}}
\affil[3]{Department of Computer Science, University of Copenhagen, Denmark, \texttt{mathias@tejs.dk}}
\affil[4]{Institute of Informatics, University of Warsaw, Poland, \texttt{kociumaka@mimuw.edu.pl}}
\date{\vspace{-1cm}}

\begin{document}

\maketitle

\begin{abstract}
A border of a string is a non-empty prefix of the string that is also a suffix of the string, and a string is unbordered if it has no border other than itself. 
Loptev, Kucherov, and Starikovskaya [CPM~2015] conjectured the following: 
If we pick a string of length $n$ from a fixed non-unary alphabet uniformly at random, then the expected maximum length of its unbordered factors is $n - O(1)$. 
We confirm this conjecture by proving that the expected value is, in fact, ${n - \Theta(\sigma^{-1})}$, where $\sigma$ is the size of the alphabet. 
This immediately implies that we can find such a maximal unbordered factor in linear time on average. However, we go further and show that the optimum average-case running time is in $\Omega (\sqrt{n}) \cap O (\sqrt{n \log_\sigma n})$ due to analogous bounds by Czumaj and Gąsieniec [CPM~2000]
for the problem of computing the shortest period of a uniformly random string.
\end{abstract}

\section{Introduction}
\label{sec:introduction}
Let $\Sigma$ be a finite \emph{alphabet} of size $\sigma\ge 2$.
A \emph{string} $S\in \Sigma^n$ is a sequence $S=S[1]\cdots S[n]$ of $n$ symbols from~$\Sigma$;
the \emph{length} $n$ of $S$ is denoted by $|S|$.
For $1\leq i\leq j\leq n$, we denote $S[i,j]=S[i]\cdots S[j]$ and call the string $S[i,j]$ a \textit{factor} of $S$.
A factor $S[1,j]$ is a \emph{prefix} of $S$ and a factor $S[i,n]$ is a \emph{suffix} of $S$.
A \textit{border} of a string is a non-empty prefix of the string that is also a suffix of the string. 
In other words, the string $S$ has a border of length $\ell$, $1\le \ell \le n$,
if and only if $S[1,\ell]=S[n-\ell+1,n]$.

A string $S$ is \emph{unbordered} if it does not have any proper border, i.e., any border other than the whole of~$S$.
By $L(S)$ we denote the maximum length of unbordered factors of $S$.
Any unbordered factor of length $L(S)$ is called a \emph{maximal unbordered factor} of $S$.

An integer $p>0$ is a \emph{period} of a string $S\in \Sigma^n$ if $S[i]=S[i+p]$ for $1\le i \le n-p$.
The shortest period of a string $S$ is denoted $\per(S)$.
Note that $p$ is a period of $S$ if and only if $S$ has a border of length $n-p$, so $S$ is unbordered if and only if $\per(S)=n$.
Moreover, $\per(S[i,j])\le \per(S)$; applied to a maximal unbordered factor, this yields $L(S)\le \per(S)$.

\begin{example}[\cite{DBLP:journals/dm/AssousP79}]
If $S=\texttt{1011001101}$, then $\per(S)=7$ and $L(S)=6$.
The maximal unbordered factors are $S[1,6]=\texttt{101100}$ and $S[5,10]=\texttt{001101}$.
\end{example}

Unbordered factors were first studied by Ehrenfeucht~and~Silberger~\cite{ehrenfeucht1979periodicity}, with emphasis on the relationship $\per(S)$ and $L(S)$. The question when $\per(S)=L(S)$ received more attention in the literature~\cite{DBLP:journals/dm/AssousP79,duval1982relationship,holub2012ehrenfeucht,harju2007periodicity}. 
For strings $S\in \Sigma^n$, the equality holds if $L(S)\le \frac37n$~\cite{holub2012ehrenfeucht} or $\per(S)\le \frac12n$~\cite{ehrenfeucht1979periodicity}.

Loptev, Kucherov, and Starikovskaya~\cite{loptev2015maximal} proved that for uniformly random string $S\in \Sigma^n$  over an alphabet $\Sigma$ of size $\sigma \geq 2$  the expected maximum length $\Exp{L(S)}$ of unbordered factors is at least ${n(1-\xi(\sigma)\cdot \sigma^{-4})}.+O(1)$, where $\xi(\sigma)$ converges to $2$ as $\sigma$ grows. 
When $\sigma \geq 5$ and $n$ is sufficiently large, their bound implies $\Exp{L(S)}\ge 0.99n$. Supported by experimental results, Loptev et al.~\cite{loptev2015maximal} conjectured that $\Exp{L(S)}=n-O(1)$. 
In Section~\ref{sec:proof}, we confirm this conjecture and prove that the tail of $n-L(S)$ decays exponentially.
\begin{theorem}
\label{thm:length}
	Let $S\in \Sigma^n$ be a uniformly random string over an alphabet $\Sigma$ of size $\sigma \ge 2$.
	\begin{enumerate}[\rm(a)]
	\item\label{it:exp} $\Exp{L(S)}=n - O(\sigma^{-1})$.
	\item\label{it:whp} For each $\delta>0$, the probability of $L(S)=n-O(\log_\sigma \delta^{-1})$ is at least $1-\delta$.
	\end{enumerate}
\end{theorem}
One can easily deduce that $\per(S)\ge L(S)$ also satisfies both claims of Theorem~\ref{thm:length}.
However, a recent study by Holub and Shallit~\cite{holub2016periods} provides much stronger results concerning the shortest periods of uniformly random strings.

The problem of computing a maximal unbordered factor of a uniformly random string was studied by Loptev~et~al.~\cite{loptev2015maximal} and Gawrychowski~et~al.~\cite{gawrychowski2015computing},
who gave algorithms with average-case running times of $O(\frac{n^2}{\sigma^4}+n)$ and $O(n\log n)$, respectively. 
The solution by Loptev et al.~\cite[Theorem 3]{loptev2015maximal} actually takes $O(n(n-L(S)+1))$ worst-case time.
By Theorem~\ref{thm:length}\eqref{it:exp}, its average-case running time is therefore $O(n)$.
Nevertheless, this is still much worse than what is necessary to compute the shortest period of a uniformly random string~\cite{DBLP:conf/cpm/CzumajG00}.
To address this issue, in Section~\ref{sec:algorithm} we develop a pair of reductions using Theorem~\ref{thm:length}\eqref{it:whp} to show that computing $L(S)$ and $\per(S)$
is equivalent with respect to the average-case running time.

\begin{theorem}\label{thm:reductions}
	Let $S\in \Sigma^n$ be a uniformly random string over an alphabet $\Sigma$ of size $\sigma$.
\begin{enumerate}[(a)]
  \item\label{it:unborderedtoperiod} The problem of computing $L(S)$ can be reduced in $O(\log_{\sigma}n)$ expected time to the problem of computing $\per(S')$ for a fixed factor $S'$ of $S$.
  \item\label{it:periodtounbordered} The problem of computing $\per(S)$ can be reduced in $O(1)$ expected time to the problem of computing $L(S)$.
\end{enumerate}
\end{theorem}

Consequently, the $\Omega (\sqrt{n})$ and $O (\sqrt{n \log_\sigma n})$ lower and upper bounds known for computing the shortest period of a uniformly random string, both due to Czumaj and Gąsieniec~\cite{DBLP:conf/cpm/CzumajG00},
carry over to computing a maximal unbordered factor of such a string.

\begin{corollary}
\label{cor:algorithm}
	The problem of computing a maximal unbordered factor of a uniformly random string over an alphabet $\Sigma$ of size $\sigma$ takes
	$O (\sqrt{n \log_\sigma n})$ time on average, and this bound is within an $O (\sqrt{\log_\sigma n})$ factor of optimal.
\end{corollary}

Czumaj and Gąsieniec also conjectured that the optimum average-case running time of computing the shortest period is $\Theta (\sqrt{n \log_\sigma n})$; 
any resolution of this conjecture automatically transfers to maximal unbordered factors.

The worst-case running time we get from Theorem~\ref{thm:reductions} and Czumaj and Gąsieniec's work~\cite{DBLP:conf/cpm/CzumajG00} is $O(n^2)$.
However, to obtain state-of-the-art running time both in the average case and in the worst case, 
we can dovetail our solution with any of the worst-case algorithms for computing a maximal unbordered factor.
Gawrychowski~et~al.~\cite{gawrychowski2015computing} gave such an algorithm with the running time $O(n^{1.5})$.
Very recently, this has been improved~\cite{DBLP:journals/corr/abs-1805-09924} to $O(n \log n \log^2 \log n)$ (and further to $O(n\log n)$ if one allows Las Vegas randomization).
Nevertheless, this is still slower than the $O(n)$ time needed to compute the shortest period in the worst-case~\cite{morris1970linear,DBLP:journals/siamcomp/KnuthMP77}.

Data structures for answering a period queries have also recently been developed. Such a query takes two indices $i$ and $j$ and the answer is the shortest period $\per(S[i,j])$. Kociumaka~et~al.~\cite{DBLP:conf/soda/KociumakaRRW15} developed a data structure of size $O(n)$ answering period queries in $O(\log n)$ time, which improved upon several earlier time-space trade-offs they presented in an earlier paper~\cite{DBLP:conf/spire/KociumakaRRW12}. Computing $L(S[i,j])$ for a given factor $S[i,j]$ appears to be a much more difficult task.

Another interesting possibility is to extend our results from average-case analysis to smoothed analysis~\cite{spielman2004smoothed,spielman2009smoothed,boucher2010large}, in which the input can be chosen adversarially but some random noise is then added to it.  We conjecture that when the noise level is reasonably large --- e.g., each symbol is replaced by a randomly chosen one with some positive constant probability --- then our bounds do not change significantly. Our results or techniques could also be applicable to other problems concerning borders and periods.   

\section{Distribution of Maximum Length of Unbordered Factors}
\label{sec:proof}
Let us fix an alphabet $\Sigma$ of size $\sigma\ge 2$.
For every $n\ge 0$, we define a random variable $\Delta_n$ distributed as $|S|-L(S)$ for uniformly random $S\in \Sigma^n$.
The following lemma, which gives a common upper bound of the \emph{moment-generating functions} $M_{\Delta_n}(t)=\Exp{e^{t\Delta_n}}$,
 is the key tool behind Theorem~\ref{thm:length}.
\begin{lemma}\label{lem:mgf}
For $n\in \mathbb{N}$ and $0\le t\le 0.1\ln\sigma$, 
we have $M_{\Delta_n}(t)\le C(t)$,
where \begin{equation}\label{eq:mgf}
C(t)= \frac{\sigma^3-\sigma^2 e^{2t}}{\sigma^3 - 2\sigma^2 e^{2t}+e^{4t}}\ .
\end{equation}
\end{lemma}
\begin{proof}
We proceed by induction on $n$.
The base case is $n\in \{0,1\}$ for which $\Delta_n = 0$ and therefore $M_{\Delta_n}(t)=1$.
Consequently, we need to prove that
\[C(t)-M_{\Delta_n}(t) = \frac{\sigma^3-\sigma^2 e^{2t}}{\sigma^3 - 2\sigma^2 e^{2t}+e^{4t}}-1 = \frac{\sigma^2 e^{2t}-e^{4t}}{\sigma^3 - 2\sigma^2 e^{2t}+e^{4t}}\ge 0.\]
Note that the denominator is a quadratic function of $e^{2t}$ with a minimum at $e^{2t}= \sigma^2$.
Hence, $\sigma^3 - 2\sigma^2 e^{2t}+e^{4t}\ge \sigma^3 - 2\sigma^{2.2}+\sigma^{0.4}$ for $t\le 0.1 \ln \sigma$. 
The right-hand side is a polynomial of $\sigma^{0.2}$, and one can easily verify that it is positive for $\sigma\ge 2$.
Consequently, the denominator is positive. 
To complete the proof of the base case, observe that $e^{2t}(\sigma^2-e^{2t})$ is also positive for $t\le \ln \sigma$.

 For $n \ge 2$, we assume $M_{\Delta_m}(t)\le C(t)$ for $m < n$ and $0\le t\le 0.1\ln \sigma$. 
 We consider a uniformly random $S\in \Sigma^n$ and condition over the possible lengths $\ell$ of the shortest border of $S$.
 More formally, we define $F(S)$ as the smallest integer $\ell>0$ such that  $S[1,\ell]=S[n-\ell+1,n]$,
 and we write
 \begin{equation}\label{eq:first}
 M_{\Delta_n}(t)= \Exp{e^{t(n-L(S))}} = \sum_{\ell=1}^n \Prb{F(S)=\ell}\cdot \Exp{e^{t(n-L(S))} \mid F(S)=\ell}.
 \end{equation}
 Now, we bound from above individual terms of this sum.
 Observe that $F(S)=n$ is equivalent to $L(S)=n$ and therefore \begin{equation}\label{eq:exp:n}\Exp{e^{t(n-L(S))} \mid F(S)=n}=1.\end{equation}
 For $\ell \le \frac12 n$, we observe that $S[\ell+1,n-\ell]$ is independent from $F(S)=\ell$.
 Due to $L(S)\ge L(S[\ell+1,n-\ell])$, this yields
 \begin{multline}\label{eq:exp:small}
 \Exp{e^{t(n-L(S))} \mid F(S)=\ell} \le \Exp{e^{t(n-L(S[\ell+1,n-\ell]))} \mid F(S)=\ell}=\Exp{e^{t(n-L(S[\ell+1,n-\ell]))} }=\\=
 e^{2t\ell}\Exp{e^{t(n-2\ell-L(S[\ell+1,n-\ell]))} }=e^{2t\ell}M_{\Delta_{n-2\ell}}(t).
 \end{multline}
 Moreover, we note that $F(S)=\ell$ implies $S[i]=S[n-\ell+i]$ for $1\le i \le \ell$ and these events are independent.
 For $\ell\ge 2$, we have one more independent event $S[1]\ne S[\ell]$ due to $F(S)\ne 1$.
 Consequently,
 \begin{equation}\label{eq:prb:small}
 \Prb{ F(S)=\ell}\le \begin{cases}
 \sigma^{-1} & \text{if }\ell=1,\\
 (\sigma-1)\sigma^{-\ell-1} & \text{if }2\le \ell \le \tfrac12n.
 \end{cases}
 \end{equation}
In the remaining case of $\tfrac12 n < \ell < n$, we observe that if $S[1,\ell]=S[n-\ell+1,n]$, then $S[n-\ell+1,\ell]$ is also a border of $S$.
This contradicts $F(S)=\ell$ because $|S[n-\ell+1,\ell]|=2\ell-n < \ell$.
 Consequently, 
 \begin{equation}\label{eq:prb:large}\Prb{F(S)=\ell}=0\quad \text{ if }\tfrac12 n < \ell < n.\end{equation}
 Plugging (\ref{eq:exp:n}--\ref{eq:prb:large}) into \eqref{eq:first}, we obtain
 \begin{align}
  \notag M_{\Delta_n}(t) &\le \Prb{F(S)=n}+\sum_{\ell=1}^{\lfloor{n/2}\rfloor} \Prb{F(S)=\ell}\cdot e^{2t\ell}\cdot M_{\Delta_{n-2\ell}}(t) \\
  &\le 1+\sigma^{-1}\cdot e^{2t}\cdot M_{\Delta_{n-2}}(t) + \sum_{\ell=2}^{\lfloor{n/2}\rfloor}(\sigma-1)\sigma^{-\ell-1}\cdot e^{2t\ell}\cdot M_{\Delta_{n-2\ell}}(t).
 \end{align}
 The inductive assumption further yields
 \begin{align}
 \notag  M_{\Delta_n}(t) &\le 1 + \sigma^{-1}\cdot e^{2t}\cdot C(t) + \sum_{\ell=2}^{\lfloor{n/2}\rfloor}(\sigma-1)\sigma^{-\ell-1}\cdot e^{2t\ell}\cdot C(t) \\
  \notag &\le 1 + C(t)\left(\sigma^{-1} e^{2t} + (\sigma-1)\sigma^{-3}e^{4t}\cdot \sum_{\ell=0}^{\infty}(\sigma^{-1} e^{2t})^\ell\right) \\ 
  \notag &=  1 + C(t)\left(\sigma^{-1} e^{2t} + (\sigma-1)\sigma^{-3}e^{4t}\cdot \frac{1}{1-\sigma^{-1}e^{2t}}\right) \\
   &= 1 + C(t)\cdot \frac{\sigma(\sigma-e^{2t})e^{2t}-(\sigma-1)e^{4t}}{\sigma^2(\sigma-e^{2t})}\\ 
  \notag    &= 1 + \frac{\sigma^3-\sigma^2 e^{2t}}{\sigma^3 - 2\sigma^2 e^{2t}+e^{4t}}\cdot \frac{\sigma^2e^{2t}-e^{4t}}{\sigma^3-\sigma^2 e^{2t}}\\
   \notag    &= \frac{\sigma^3 - 2\sigma^2 e^{2t}+e^{4t}\sigma^2e^{2t}-e^{4t}}{\sigma^3 - 2\sigma^2 e^{2t}+e^{4t}}\\
   \notag &= C(t).
 \end{align}
 This completes the proof of Lemma~\ref{lem:mgf}.
 \end{proof}
 
 Next, let us focus on the expected value $\Exp{\Delta_n}$.
 Note that $M_{\Delta_n}(t)=\Exp{e^{t\Delta_n}}\ge \Exp{1+t\Delta_n}$.
 Consequently, for $0< t\le 0.1\ln \sigma$ we have
 \begin{equation}
 \Exp{\Delta_n} \le \frac{M_{\Delta_n}(t)-1}{t}\le \frac{C(t)-1}{t}\ .
 \end{equation}
 Hence, $\Exp{\Delta_n}$ is bounded by a function of $\sigma$ independent of $n$.
 To analyze its asymptotics in terms of $\sigma$, we plug $t=1$ (valid for $\sigma \ge e^{10}$), which yields
  \begin{equation}
 \Exp{\Delta_n} \le C(1)-1 = \frac{\sigma^2  e^2-e^4}{\sigma^3 - 2\sigma^2 e^2 +e^4} = \frac{O(\sigma^2)}{\Omega(\sigma^3)}=O(\sigma^{-1}).
 \end{equation}
 This completes the proof of Theorem~\ref{thm:length}\eqref{it:exp}.
 
 For the claim~\eqref{it:whp}, we apply Markov's inequality on top of Lemma~\ref{lem:mgf}:
 \begin{equation}
 \Prb{\Delta_n \ge \ell} \le \frac{\Exp{e^{t\Delta_n}}}{e^{t\ell}} = \frac{M_{\Delta_n}(t)}{e^{t\ell}}\le \frac{C(t)}{e^{t\ell}}.
 \end{equation}
Hence, it suffices to take $\ell\ge 10 \log_{\sigma}(\delta^{-1}\cdot C(0.1\ln \sigma))$ to make sure that the probability does not exceed $\delta$.
To complete the proof, observe that
\begin{equation}
C(0.1\ln \sigma)=\frac{\sigma^3-\sigma^{2.2}}{\sigma^3 - 2\sigma^{2.2}+\sigma^{0.4}}=\frac{O(\sigma^3)}{\Omega(\sigma^3)}=O(1).
\end{equation}
 
 \section{Average-Case Algorithms for Maximal Unbordered Factors}\label{sec:algorithm}
 
In this section, we give a pair of reductions between the problems of computing the shortest period and the maximum length of unbordered factors
of a uniformly random string, thereby proving Theorem~\ref{thm:reductions}.
We assume that the alphabet $\Sigma$ is of size $\sigma\ge 2$. Otherwise, both values are always 1.

We start with a simple argument showing Theorem~\ref{thm:reductions}\eqref{it:periodtounbordered}.
Suppose that we aim at computing $\per(S)$ for a uniformly random string $S\in \Sigma^n$. 
Having determined $L(S)$, we rely on the fact that $\per(S)\ge L(S)$. 
We construct a string $\Sd := S[1,n-L(S)]\$ S[L(S)+1,n]$, where $ \$\notin \Sigma$ is a sentinel symbol,
and observe that $S$ has a border of length $\ell \le n-L(S)$ if and only if $\Sd$ has such a border.
Moreover, the presence of the sentinel symbol guarantees that $\Sd$ does not have proper borders longer than $n-L(S)$.
Consequently, we have $|S|-\per(S) = |\Sd|-\per(\Sd)$. The value $\per(\Sd)$ can be computed using a worst-case algorithm~\cite{morris1970linear,DBLP:journals/siamcomp/KnuthMP77},
which takes $O(|\Sd|)=O(n-L(S)+1)$ time. The expected running time of the reduction is $O(1)$ due to Theorem~\ref{thm:length}\eqref{it:exp}.

We proceed with a proof of Theorem~\ref{thm:reductions}\eqref{it:unborderedtoperiod}.
Suppose that we aim at computing $L(S)$ for a uniformly random string $S\in \Sigma^n$.
We apply Theorem~\ref{thm:length}\eqref{it:whp} for $\delta = \frac{1}{n^2}$ to obtain a value $d=O(\log_{\sigma }n)$
such that $\Prb{|T|-L(T)\ge d} \le \frac{1}{n^2}$ for uniformly random strings $T \in \Sigma^{m}$ of arbitrary length $m$.
Note that this also yields $\Prb{|T|-\per(T)\ge d} \le \frac{1}{n^2}$ due to $\per(T)\ge L(T)$.

If $n \le 6d$, we simply determine $L(S)$ using Loptev et al.'s algorithm~\cite{loptev2015maximal}, which takes $O(d)=O(\log_{\sigma }n)$ time on average.
Otherwise, we construct three strings
\begin{align*}
\Sb &:= S[1,3d]S[n-3d+1,n],\\
S' &:= S[d+1,n-d],\\ 
 \Sb' &:= S[d+1,3d]S[n-3d+1,n-d],
\end{align*}
and we compute $|\Sb|-L(\Sb)$, $|S'|-\per(S')$, and $|\Sb'|-\per(\Sb')$.
If any of these values exceeds $d$, we fall back to the algorithm of~\cite{loptev2015maximal}
to compute $L(S)$. Otherwise, we determine $L(S)$ based on $|S|-L(S)=|\Sb|-L(|\Sb|)$.

Before proving this equality, let us analyze the running time of the reduction. 
Observe that $\Sb$, $S'$, and $\Sb'$ are uniformly random strings of the respective lengths, which lets us use average-case algorithms.
In particular, it takes $O(d)$ time on average to compute $L(\Sb')$ using Loptev et al.'s algorithm~\cite{loptev2015maximal}.
Determining $\per(S')$ is the target of the reduction, so we do not include it in the analysis.
The value $\per(\Sb')$ is computed in $O(d)$ worst-case time~\cite{morris1970linear,DBLP:journals/siamcomp/KnuthMP77}.
The probability of a fall-back is at most $\frac3{n^2}$ by the choice of $d$,
which compensates for the worst-case\footnote{Note that we cannot use the average-case bound of $O(n)$ because the conditional distribution of $S$ (in case of a fall-back) is no longer uniform across $\Sigma^n$.} time $O(n^2)$ it takes to apply Loptev et al.'s algorithm to the whole of $S$.
Overall, the reduction works in $O(d)=O(\log_\sigma n)$ time on average.

It remains to prove $|S|-L(S)=|\Sb|-L(\Sb)$ provided that $|\Sb|-L(\Sb)\le d$, $|S'|-\per(S')\le d$, and $|\Sb'|-\per(\Sb')\le d$.
First, consider a maximal unbordered factor of $\Sb$. It must be of the form $S[i,3d]S[n-3d+1,j]$ for some $1\le i \le d$ and $n-d+1\le j \le n$, and we claim that $S[i,j]$ is then an unbordered factor of $S$. 
For a proof by contradiction, suppose that $S[i,j]$ has a proper border and the longest such border is of length $\ell$.
Note that $\ell>\min(|S[i,3d]|,|S[n-3d+1,j]|)$ because $S[i,3d]S[n-3d+1,j]$ is unbordered.
We conclude that $\per(S[i,j])= |S[i,j]|-\ell < n-3d$. However, this yields $\per(S')\le \per(S[i,j])< n-3d = |S'|-d$, a contradiction.
Consequently, $|S|-L(S)\le |\Sb|-L(\Sb)$.

The proof of $|S|-L(S)\ge |\Sb|-L(\Sb)$ is symmetric.
 We consider a maximal unbordered factor $S[i,j]$ of $S$, observe that $1\le i \le d$ and $n-d+1\le j \le n$ due to $|S|-L(S)\le d$,
 and claim that $S[i,3d]S[n-3d+1,j]$ is unbordered
For a proof by contradiction we suppose that it a border of length $\ell$. We note that  $\ell>\min(|S[i,3d]|,|S[n-3d+1,j]|)$ because $S[i,j]$ is unbordered
and derive $\per(\Sb')\le \per(S[i,3d]S[n-3d+1,j])< 3d$, which contradicts $\per(\Sb')\ge |\Sb'|-d = 3d$.

This completes the proof of Theorem~\ref{thm:reductions}\eqref{it:unborderedtoperiod}.

\section*{Acknowledgments}

Many thanks to Danny Hucke for asking about the possibility of a sublinear average-case algorithm at the presentation of the conference version of this paper, and to the anonymous reviewers for their comments. %, including the observation that we can dovetail our algorithm and Gawrychowski~et~al.'s.

\bibliographystyle{plainurl} 
\bibliography{unbordered_arXiv}

\end{document}